\documentclass[preprint,11pt]{elsarticle}
\usepackage[letterpaper,margin=1in]{geometry}
\usepackage[T1]{fontenc}
\usepackage{amsmath,amssymb,amsfonts,mathtools,amsthm}
\usepackage{booktabs,array,longtable,tabularx,graphicx}
\usepackage{tikz}
\usetikzlibrary{arrows.meta,calc,positioning}
\usepackage{microtype}
\usepackage[hidelinks]{hyperref}
\hypersetup{pdftitle={Exact and Asymptotic Analysis of Greedy Online Facility Assignment on Cycle Metrics},
pdfauthor={Md. Tanzeem Rahat, Md. Rawha Siddiqi Riad, Md. Manzurul Hasan}}
\newcolumntype{Y}{>{\raggedright\arraybackslash}X}
\theoremstyle{plain}
\newtheorem{theorem}{Theorem}
\newtheorem{lemma}[theorem]{Lemma}
\newtheorem{proposition}[theorem]{Proposition}
\newtheorem{corollary}[theorem]{Corollary}
\theoremstyle{definition}

\theoremstyle{remark}

\DeclareMathOperator{\dist}{dist}
\newcommand{\E}{\mathbb{E}}

\newcommand{\V}{\mathcal{V}}

\newcommand{\OPT}{\mathrm{OPT}}

\newcommand{\orcidid}[1]{\href{https://orcid.org/#1}{\textsuperscript{\normalfont\scriptsize ORCID}}}
\journal{Theoretical Computer Science}
\begin{document}
\begin{frontmatter}
\title{Exact and Asymptotic Analysis of Greedy Online Facility Assignment on Cycle Metrics}
\author[aiub]{Md. Tanzeem Rahat\fnref{orcid-rahat}}
\ead{tanzeem.rahat@aiub.edu}
\author[aiub]{Md. Rawha Siddiqi Riad}
\ead{22-46678-1@student.aiub.edu}
\author[aiub]{Md. Manzurul Hasan\fnref{orcid-hasan}\corref{cor}}
\ead{manzurul@aiub.edu}
\cortext[cor]{Corresponding author.}
\fntext[orcid-rahat]{ORCID: \href{https://orcid.org/0009-0008-6474-9877}{0009-0008-6474-9877}.}
\fntext[orcid-hasan]{ORCID: \href{https://orcid.org/0009-0003-8495-1013}{0009-0003-8495-1013}.}
\address[aiub]{Department of Computer Science and Engineering,
American International University-Bangladesh,
408/1, Kuratoli, Khilkhet, Dhaka 1229, Bangladesh}
\begin{abstract}
We identify the free-facility gaps of greedy online assignment on a cycle
with the discrete additive coalescent and use this connection to determine the
exact expected assignment cost. The cycle has circumference $n$, with one
unit-capacity facility at each integer point and $n$ independent uniform
arrivals. Each request is assigned immediately to a nearest free facility.
The conditional one-step cost is the sum of squared gap lengths divided by
$4n$. A uniform cyclic-order argument identifies the unordered gap process
with additive coalescence. The classical Cayley-tree representation then
reduces the accumulated cost to a two-vertex distance functional, yielding
\[
G_n=\frac n4\sum_{d=0}^{n-1}\prod_{j=1}^{d}\left(1-\frac jn\right)
\sim\frac14\sqrt{\frac\pi2}\,n^{3/2}.
\]
The formula requires a linear number of scalar recurrence steps to evaluate.
We also obtain an asymptotic upper bound of $\pi\sqrt3/2$ on the ratio of
expected greedy cost to expected offline optimum under this arrival model.
An exact rational Bellman calculation shows that greedy is optimal for three
facilities but exceeds the optimal online expected cost by $1/512$ for four
facilities. Independent exact recurrences and finite tree enumeration provide
reproducible checks of the analysis.
\end{abstract}
\end{frontmatter}

\section{Introduction}\label{sec:intro}

Nearest-available greedy assignment has a useful structural description on
a cycle: consuming a facility merges the two gaps on its sides. Under uniform
independent arrivals, the removal probability is proportional to the sum of
those gap lengths. We show that, after random labeling and averaging over
cyclic order, this process is the discrete additive coalescent. The connection
allows us to compute the total expected assignment cost exactly by using a
distance functional of a uniform labeled tree.

The assignment problem has $n$ unit-capacity facilities equally spaced on a
cycle of circumference $n$. Exactly $n$ requests arrive independently and
uniformly on the cycle, and each is assigned irrevocably to a nearest free
facility in cycle distance. Our quantity of interest is the expected total
cost $G_n$ of this fixed policy. Capacity depletion creates long random gaps,
so later requests can travel much farther than early ones. Tracking the gaps
makes both this geometric effect and its contribution to cost explicit.

\subsection{Main result and proof strategy}

We prove, for every $n\geq3$,
\begin{equation}
G_n=\frac n4\sum_{d=0}^{n-1}
\prod_{j=1}^{d}\left(1-\frac jn\right),
\qquad
G_n\sim\frac14\sqrt{\frac\pi2}\,n^{3/2},
\label{eq:intro-exact}
\end{equation}
where the empty product is one. Thus the average cost per arrival has order
$\sqrt n$, with an explicit leading constant. The formula also replaces a
recurrence over occupied subsets by a scalar recurrence with $n-1$ updates
and a length-$n$ sum; this is a count of arithmetic steps, not a claim about
exact bit complexity.

The proof has two model-specific ingredients. First, when the free facilities
determine gaps $g_1,\ldots,g_k$, the conditional immediate cost is
$\sum_i g_i^2/(4n)$, and removing a boundary between gaps of lengths $x,y$
has probability $(x+y)/(2n)$. Second, conditional on the randomly labeled
gap partition, its oriented cyclic order remains uniform. This yields the
unordered merger probability $(x+y)/(n(k-1))$ of the additive coalescent.
Using its classical random-tree representation, we then express the sum of
expected squared gap masses as an integrated connectivity probability. If
$D$ is the distance between two independent uniform vertices of a uniform
Cayley tree, allowing the vertices to coincide, this gives
\[
G_n=\frac{n^2}{4}\E[(D+1)^{-1}].
\]
The known Cayley-tree distance law yields the finite formula, and a Gaussian
Riemann-sum estimate yields its asymptotic behavior.

This argument uses established coalescent and tree constructions
\citep{aldous1998standard,pitman1999coalescent}, with a related finite
coalescent already present in cyclic parking
\citep{chassaing2002parking}. The contribution here is the identification of
the nearest-free geometric assignment process and the evaluation of its
accumulated distance cost. Section~\ref{sec:related-coalescent} makes this
relationship precise.

\subsection{Stochastic comparisons}

Two further results place the exact greedy cost in context. A Lipschitz
certificate for offline matching gives
\[
\limsup_{n\to\infty}
\frac{G_n}{\E[\OPT_n^{\mathrm{off}}]}
\leq\frac{\pi\sqrt3}{2}.
\]
This compares expectations under the specified arrival distribution.
For the optimal online policy, exact rational Bellman integration gives
$V_3^\star=G_3=17/12$ and $V_4^\star=1135/512$, whereas $G_4=1136/512$.
The square therefore provides an explicit strict improvement over greedy.
These finite calculations leave the asymptotic online-optimality question
open.

\subsection{Relationship to earlier work}\label{sec:earlier-work}

An earlier study by the authors \citep{riad2026isco} considered the same
stochastic cycle model using occupied-set recurrences and symmetry-based
computation. The present paper develops a coalescent analysis that yields an
exact expected-cost formula and sharp asymptotics, together with stochastic
performance comparisons. Table~\ref{tab:earlier-work} summarizes the shared
material and the results established here. The earlier occupied-set recurrence
is retained in \ref{app:recurrences} for independent verification;
\ref{app:numerical-correction} records the correction to the earlier
large-instance numerical estimates.

\begin{table}[htbp]
\centering
\caption{Comparison with the authors' earlier ISCO study
\citep{riad2026isco}, using its public version arXiv:2512.00506v3.}
\label{tab:earlier-work}
\small
\renewcommand{\arraystretch}{1.15}
\begin{tabularx}{\textwidth}{@{}>{\raggedright\arraybackslash}p{0.20\textwidth}YY@{}}
\toprule
Aspect & Earlier ISCO study & Present paper\\
\midrule
Stochastic cycle model & Introduced & Retained\\
Analytical approach & Occupied-set recurrence and dihedral symmetry reduction
& Gap process, additive coalescence, and Cayley trees\\
Expected greedy cost & Finite-instance calculations, including $G_4=71/32$
& Exact formula for every $n\geq3$\\
Asymptotic analysis & No sharp asymptotic theorem
& $G_n\sim\frac14\sqrt{\pi/2}\,n^{3/2}$\\
Stochastic benchmarks & No corresponding results
& Offline comparison and optimal-online Bellman analysis\\
Numerical results & Preliminary computations
& Corrected large-instance values and independent verification\\
\bottomrule
\end{tabularx}
\end{table}

\paragraph{Organization}
Section~\ref{sec:related} discusses related work. Sections~\ref{sec:model}
and~\ref{sec:gaps} define the model and derive its gap dynamics.
Sections~\ref{sec:coalescent}--\ref{sec:asymptotic} prove the structural
representation, exact formula, and asymptotic theorem.
Sections~\ref{sec:benchmarks} and~\ref{sec:optimal-online} give the stochastic
comparisons. Section~\ref{sec:computation} summarizes reproducibility, and
Section~\ref{sec:discussion} discusses limitations and open problems.
Appendices contain the validation recurrences, numerical checks, and the
complete Bellman certificate.

\section{Related Work}\label{sec:related}

\subsection{Additive coalescence, parking, and random trees}\label{sec:related-coalescent}

The additive coalescent uses the merger kernel $K(x,y)=x+y$.
\citet{aldous1998standard} developed the standard additive coalescent and its
scaling theory. For the finite process, \citet[Theorem~5]{pitman1999coalescent}
identified the equivalence between a uniform Cayley tree with a uniform edge
order and the merger probability $(x+y)/(n(k-1))$. We use this existing
representation and include its finite counting proof to fix normalization.
The forest-extension count and tree-distance law used below are also
classical \citep{moon1970counting,meir1970distance}.

The closest geometric precedent is cyclic parking with linear probing.
\citet[Eq.~(1.5) and Section~8]{chassaing2002parking} show that parking-block
masses, with an empty place included in each block, follow the same finite
additive merger rule. Thus the parking connection holds at finite size as
well as in the scaling limit. Their directed parking rule and block-size
analysis provide the relevant structural precedent for our nearest-free
assignment model.

Here the cost is shortest cycle distance from a continuous arrival to its
assigned facility. Its conditional expectation is a specific second-moment
functional of the free gaps, $\sum_i g_i^2/(4n)$. We prove the cyclic-order
invariant for this geometric rule and sum that functional over all occupancy
levels through two-vertex connectivity. This identifies precisely the
additional cost calculation needed to obtain $G_n$ from the classical
coalescent representation. Our novelty claim concerns this assignment
analysis; the coalescent, its tree representation, and the Cayley distance
distribution are established ingredients.

\subsection{Online transportation and metric matching}
In online transportation, requests arrive in a metric space and must be
assigned irrevocably to server sites with residual capacity.
\citet{kalyanasundaram2000transportation} studied this problem under resource
augmentation and analyzed greedy and balance-based rules.
The recent algorithm of \citet{harada2025transportation} gives the first
deterministic $O(m)$-competitive guarantee when there are $m$ server sites.
\citet{chung2008transportation} showed how one extra server per site can
produce polylogarithmic guarantees on hierarchically separated trees and,
through metric embeddings, randomized guarantees on general metrics. These
works precede the OFA terminology and contain the same fundamental combination
of irrevocability, metric cost, fixed sites, and capacity.

Online minimum-cost matching is the unit-capacity special case. Classical
work established strong worst-case limitations for nearest-available greedy
\citep{kalyanasundaram1993online}. Randomized algorithms for general metrics
and algorithms exploiting doubling structure were developed by
\citet{meyerson2006randomized} and \citet{gupta2012doubling}, respectively.
These results concern guarantees over adversarial request sequences, whereas
our theorem evaluates one fixed policy under a fully specified arrival law.

\subsection{Random-order and stochastic arrivals}
The random-order model fixes a request multiset and exposes it in a uniformly
random permutation; the i.i.d. stochastic model instead draws each request
independently from a distribution. These are distinct relaxations of the
adversarial model; \citet{gupta2020randomorder} survey the random-order
framework. For minimum-cost metric matching, \citet{gairing2019greedy}
analyzed nearest-free greedy under random arrival order, and
\citet{duppala2022uniform} studied randomized greedy on the uniform metric.
\citet{gupta2019stochastic} considered i.i.d. requests drawn from a known
distribution and designed competitive algorithms for general and structured
metrics. Earlier average-case work of \citet{tsai1994average} treated greedy
online minimum matching in Euclidean space. More recently,
\citet{balkanski2023power} obtained model-sensitive guarantees and lower bounds
for greedy minimum-cost matching on the line under random and semi-random
inputs. Our model belongs to the i.i.d. category, but has a deterministic
cyclic server set and a continuous uniform request distribution; our goal is
an exact all-$n$ evaluation of greedy rather than a worst-case guarantee over a
class of metric instances.

\subsection{Online facility assignment}
\citet{ahmed2020online} introduced the OFA formulation and terminology for
fixed capacitated facilities on line and graph metrics, with worst-case
competitive analysis. Line-based OFA with arbitrary server layouts and
capacities was developed further by
\citet{harada2023general,harada2024capacity}. These adversarial analyses and
the present distributional analysis address different questions: here the
facility locations are equally spaced on a cycle, the requests are i.i.d.
uniform points, and the objective is the exact expected cost of a fixed policy.

\section{Model}\label{sec:model}

Assume throughout that $n\geq3$. Let
\[
\V=\{0,1,\ldots,n-1\}
\]
be the vertices of the cycle $C_n$, with indices interpreted modulo $n$.
Every edge has length one, so the cycle has circumference $n$. Vertex
$v\in\V$ hosts one unit-capacity facility.

The discrete cycle distance between vertices is
\[
\delta(u,v)=\min\{|u-v|,\,n-|u-v|\}.
\]
A point on edge $e=\{e,e+1\}$ is parameterized by $x=e+t$ with
$t\in[0,1]$. Its distance to facility $v$ is
\begin{equation}
d(e,t;v)=
\min\{t+\delta(e,v),\,1-t+\delta(e+1,v)\}.
\label{eq:point-distance}
\end{equation}

There are exactly $n$ arrivals $X_1,\ldots,X_n$. For every arrival,
the edge is uniform over the $n$ edges and, conditional on the edge, $t$ is
uniform on $[0,1]$. Equivalently, the arrivals are independent uniform points
on the circle of length $n$.

Let $S\subseteq\V$ be the occupied facilities and $F(S)=\V\setminus S$ the
free facilities. Greedy assigns $(e,t)$ to a facility minimizing
$d(e,t;v)$ over $v\in F(S)$. Ties are broken uniformly. Since equal-distance
boundaries form a finite set on each edge, tie-breaking has no effect on any
expectation in this paper.

Write $G_n$ for the expected total greedy assignment cost from the
initial state. The occupied-state recurrence from the earlier study is
recorded in \ref{app:recurrences}; the analysis below uses gaps.

\section{The Gap Process}\label{sec:gaps}

Suppose $k$ facilities remain free and list them in clockwise order as
$f_1,\ldots,f_k$. Let $g_i$ be the clockwise arc length from $f_i$ to
$f_{i+1}$, with $f_{k+1}=f_1$. Then
\begin{equation}
g_i\in\mathbb{Z}_{\geq1},
\qquad
\sum_{i=1}^k g_i=n.
\label{eq:gap-composition}
\end{equation}
For $k=1$, the unique gap has length $n$.

\begin{theorem}[One-step gap cost]\label{thm:gap-cost}
Conditional on the current gap vector
$\mathbf g=(g_1,\ldots,g_k)$, the expected cost of the next greedy assignment is
\begin{equation}
c(\mathbf g)=\frac{1}{4n}\sum_{i=1}^k g_i^2.
\label{eq:gap-cost}
\end{equation}
\end{theorem}

\begin{proof}
First suppose $k\geq2$. The interior of a gap of length $g_i$ has no free
facility, and its two distinct endpoints are free. They are therefore the
nearest free facilities to every point in that gap. At distance
$x\in[0,g_i]$ from the clockwise initial endpoint, greedy pays
$\min\{x,g_i-x\}$. Thus the integral of the cost over that gap is
\[
\int_0^{g_i}\min\{x,g_i-x\}\,dx
=2\int_0^{g_i/2}x\,dx=\frac{g_i^2}{4}.
\]
Summing and dividing by the circumference $n$ gives
\eqref{eq:gap-cost}. If $k=1$, both endpoints of the unique length-$n$ arc
represent the same free facility, and its two directions around the cycle
still give cost $\min\{x,n-x\}$. The same integral is $n^2/4$, so the formula
also holds in the terminal state.
\end{proof}

\begin{theorem}[Adjacent gap merger]\label{thm:gap-transition}
For $k\geq2$, the free facility $f_i$ between gaps $g_{i-1}$ and $g_i$ is
consumed next with probability
\begin{equation}
p_i(\mathbf g)=\frac{g_{i-1}+g_i}{2n}.
\label{eq:gap-transition}
\end{equation}
Afterwards the two gaps are replaced by one gap of length
$g_{i-1}+g_i$.
\end{theorem}

\begin{proof}
The Voronoi region of $f_i$ extends halfway into each adjacent gap, so its
length is $(g_{i-1}+g_i)/2$. Uniform arrival on a circle of length $n$ gives
\eqref{eq:gap-transition}. Removing $f_i$ leaves no free facility at the
common boundary, so the two arcs become one.
\end{proof}

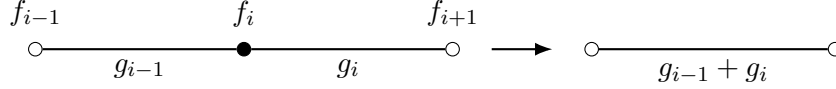
\begin{figure}[t]
\centering
\begin{tikzpicture}[x=1.15cm,y=1cm,>=Latex]
\node[circle,draw,inner sep=1.8pt] (a) at (0,0) {};
\node[circle,draw,fill=black,inner sep=1.8pt] (b) at (2.4,0) {};
\node[circle,draw,inner sep=1.8pt] (c) at (4.8,0) {};
\draw[thick] (a)--(b)--(c);
\node[above=2pt of a] {$f_{i-1}$};
\node[above=2pt of b] {$f_i$};
\node[above=2pt of c] {$f_{i+1}$};
\node[below] at (1.2,0) {$g_{i-1}$};
\node[below] at (3.6,0) {$g_i$};
\draw[->,thick] (5.25,0)--(5.95,0);
\node[circle,draw,inner sep=1.8pt] (d) at (6.4,0) {};
\node[circle,draw,inner sep=1.8pt] (e) at (9.2,0) {};
\draw[thick] (d)--(e);
\node[below] at (7.8,0) {$g_{i-1}+g_i$};
\end{tikzpicture}
\caption{Consuming the free facility $f_i$ merges its two adjacent gaps. The
black fill marks the facility consumed at the current step.}
\label{fig:gap-merge}
\end{figure}

The same formulas give an exact gap-state recurrence, recorded in
\ref{app:recurrences}. We next identify the unordered gap process.

\section{Identification with the Additive Coalescent}\label{sec:coalescent}

We now forget the cyclic order and retain only the gap masses. To distinguish
equal masses, assign the labels $1,\ldots,n$ to the initial unit gaps by a
uniform random bijection. This random labeling does not change the unlabeled gap process.
At later times, a gap is a block of labels and its mass is the block size.

\begin{theorem}[Additive-coalescent transition]\label{thm:additive}
Conditional on a partition into $k\geq2$ labeled blocks with masses
$s_1,\ldots,s_k$, a specified pair of distinct blocks of masses $x$ and $y$
is the next pair merged with probability
\begin{equation}
\frac{x+y}{n(k-1)}.
\label{eq:additive-transition}
\end{equation}
Consequently, the unordered gap-size process is the discrete additive
coalescent started from $n$ unit masses.
\end{theorem}

\begin{proof}
An \emph{oriented cyclic order} is taken modulo rotation but not reflection;
there are $(k-1)!$ such orders of $k$ distinguishable blocks. We first prove
the invariant that, conditional on the current labeled partition, its oriented
cyclic order is uniform. At time zero this follows from the uniform bijection
between the labels and the fixed cyclic positions of the unit gaps.

Assume the invariant at a level with $k\geq3$ blocks. Among the $(k-1)!$
oriented cyclic orders, two specified blocks are adjacent in exactly
$2(k-2)!$ orders: contract the pair and choose which of its two members comes
first. Thus
\[
 \Pr(\text{the specified blocks are adjacent}\mid\text{partition})
 =\frac{2(k-2)!}{(k-1)!}=\frac{2}{k-1}.
\]
If their masses are $x$ and $y$, they have one common boundary whenever they
are adjacent, and Theorem~\ref{thm:gap-transition} selects that boundary with
probability $(x+y)/(2n)$. Multiplication gives
$(x+y)/(n(k-1))$. If $k=2$, the two blocks share both cyclic boundaries; the
two corresponding removal probabilities sum to
$2(x+y)/(2n)=1$, which is again \eqref{eq:additive-transition} because
$x+y=n$ and $k-1=1$.

It remains to verify the invariant after a merger. For $k\geq3$, fix a
predecessor partition and a pair to be merged. Every oriented cyclic order of
the resulting $k-1$ blocks has exactly two predecessor orders, obtained by
expanding the merged block as $(A,B)$ or $(B,A)$. Both predecessor orders have
the same conditional probability $1/(k-1)!$, and their common boundary has
the same selection weight $(x+y)/(2n)$. Hence, conditional on this predecessor
partition and this merger, every resulting order has equal probability.
Conditioning only on the new partition merely mixes such uniform
distributions, so it remains uniform. The transition from two blocks to one
is trivial. Induction proves both the invariant and the stated transition at
every level.
\end{proof}

The transition in \eqref{eq:additive-transition} is the embedded jump chain of
the additive Marcus--Lushnikov coalescent, whose kernel is $K(x,y)=x+y$.
Random-forest constructions of the same chain are classical
\citep{pitman1999coalescent}. We next give the finite construction needed for
our cost calculation.

\begin{lemma}[Cayley forest extension {\citep{moon1970counting,pitman1999coalescent}}]\label{lem:forest-extension}
Let $F$ be a fixed forest on $[n]=\{1,\ldots,n\}$ with $k$ components of
sizes $s_1,\ldots,s_k$. The number of labeled trees on $[n]$ that contain
every edge of $F$ is
\begin{equation}
n^{k-2}\prod_{i=1}^k s_i.
\label{eq:forest-extension}
\end{equation}
\end{lemma}

\begin{proof}
If $k=1$, the forest $F$ is already a tree. There is exactly one completion,
and the right-hand side is $n^{-1}s_1=1$. Now assume $k\geq2$.
Contract the components of $F$. For a fixed tree $R$ on the $k$ contracted
vertices, an edge incident with contracted vertex $i$ has $s_i$ choices for
its endpoint in that component. Thus $R$ has
$\prod_i s_i^{\deg_R(i)}$ lifts. The weighted Cayley identity
\[
\sum_R\prod_{i=1}^k s_i^{\deg_R(i)}
=\left(\prod_{i=1}^k s_i\right)
\left(\sum_{i=1}^k s_i\right)^{k-2}
\]
follows directly from the Pr\"ufer code. Since $\sum_i s_i=n$, summing the
lifts gives \eqref{eq:forest-extension}.
\end{proof}

\begin{theorem}[Uniform-tree representation {\citep[Theorem~5]{pitman1999coalescent}}]\label{thm:tree-representation}
Choose a uniform labeled Cayley tree $T$ on $[n]$ and independently order its
$n-1$ edges uniformly. Starting with $n$ singleton components, add the edges
in that order. The resulting component-size chain has transition
\eqref{eq:additive-transition}, and hence has the same law as the unordered
greedy gap-size chain.
\end{theorem}

\begin{proof}
Suppose that the first $n-k$ added edges form a particular forest $F$ with
$k\geq2$ components. For every tree $T$ containing $F$, the probability that
the first $n-k$ positions of an independent uniform edge order contain exactly
$E(F)$ is
$\binom{n-1}{n-k}^{-1}$; for a tree not containing $F$ it is zero. Since the
prior on $T$ is uniform, the conditional completion is therefore uniform
among all trees containing $F$.

Let $A$ and $B$ be components of sizes $x$ and $y$, and fix one of the $xy$
possible vertex pairs joining them. By Lemma~\ref{lem:forest-extension}, the
ratio of the numbers of completions containing $F$ plus this candidate edge
and containing $F$ alone is
\[
\frac{n^{k-3}(x+y)\prod_{C\neq A,B}|C|}
{n^{k-2}xy\prod_{C\neq A,B}|C|}
=\frac{x+y}{nxy}.
\]
No tree completion can contain two distinct candidate edges between $A$ and
$B$, because those two edges together with paths inside $A$ and $B$ would form
a cycle. The $xy$ candidate-edge events are thus mutually exclusive, and a
uniform completion connects $A$ directly to $B$ with probability
$(x+y)/n$.

Conditional on a completion $T$ and on $E(F)$ occupying the first $n-k$
positions, each of the remaining $k-1$ edges is equally likely to be next.
Consequently the next merger is $A,B$ with probability
$(x+y)/(n(k-1))$. This probability depends on $F$ only through its component
sizes. Averaging over all forests that realize a given labeled partition
therefore gives exactly \eqref{eq:additive-transition}. The initial partitions
are the same, so induction proves equality in law of the two component-size
chains.
\end{proof}

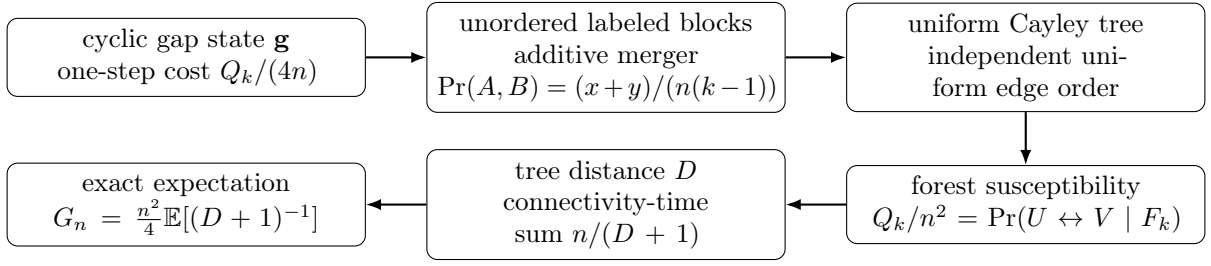
\begin{figure}[t]
\centering
\begin{tikzpicture}[
  node distance=7mm and 8mm,
  box/.style={draw,rounded corners,align=center,inner sep=4pt,
              text width=0.27\textwidth,font=\small},
  arrow/.style={-{Latex[length=2mm]},thick}
]
\node[box] (gap) {cyclic gap state $\mathbf g$\\one-step cost
$Q_k/(4n)$};
\node[box,right=of gap] (coal) {unordered labeled blocks\\additive merger
$\Pr(A,B)=(x+y)/(n(k-1))$};
\node[box,right=of coal] (tree) {uniform Cayley tree\\independent uniform
edge order};
\node[box,below=of tree] (conn) {forest susceptibility\\
$Q_k/n^2=\Pr(U\leftrightarrow V\mid F_k)$};
\node[box,left=of conn] (dist) {tree distance $D$\\connectivity-time sum
$n/(D+1)$};
\node[box,left=of dist] (formula) {exact expectation\\
$G_n=\frac{n^2}{4}\E[(D+1)^{-1}]$};
\draw[arrow] (gap)--(coal);
\draw[arrow] (coal)--(tree);
\draw[arrow] (tree)--(conn);
\draw[arrow] (conn)--(dist);
\draw[arrow] (dist)--(formula);
\end{tikzpicture}
\caption{Logical route from geometric greedy assignments to the exact formula.
The top-row arrows identify the unordered gap-size, additive-coalescent, and
Cayley-forest component-size chains in law; the lower arrows are exact
conditional-probability and expectation identities. No step is asymptotic.}
\label{fig:proof-roadmap}
\end{figure}

\section{Exact Expected Cost}\label{sec:exact}

Let $Q_k$ be the sum of squared component sizes when $k$ gaps remain:
\[
Q_k=\sum_{i=1}^k g_i^2.
\]
Theorem~\ref{thm:gap-cost} gives
\begin{equation}
G_n=\frac{1}{4n}\sum_{k=1}^n\E[Q_k].
\label{eq:cost-susceptibility}
\end{equation}

In the tree representation, independently choose vertices $U$ and $V$
uniformly from $[n]$, allowing $U=V$, and let $D=\dist_T(U,V)$.
Conditional on a forest with component sizes $g_1,\ldots,g_k$,
\[
\Pr(U\leftrightarrow V\mid\text{forest})
=\frac{Q_k}{n^2}.
\]

\begin{lemma}[Connectivity sum]\label{lem:connectivity-sum}
For $0\leq d\leq n-1$, conditional on $D=d$,
\begin{equation}
\sum_{k=1}^n
\Pr(U\leftrightarrow V\text{ when }k\text{ components remain}\mid D=d)
=\frac{n}{d+1}.
\label{eq:connectivity-sum}
\end{equation}
\end{lemma}

\begin{proof}
When $k$ components remain, the last $k-1$ edges in the uniform ordering are
absent. Because the edge order is independent of $T,U,V$, conditional on all
three objects these absent edges form a uniform $(k-1)$-subset of the $n-1$
tree edges. The vertices $U,V$ are connected exactly when none of the $d$
edges on their unique tree path is absent. Therefore the conditional
probability is
\[
\frac{\binom{n-1-d}{k-1}}{\binom{n-1}{k-1}},
\]
with the numerator interpreted as zero when necessary. Put $N=n-1$ and
$r=k-1$. The beta integral
\[
\binom{N}{r}^{-1}
=(N+1)\int_0^1 t^r(1-t)^{N-r}\,dt
\]
gives
\begin{align*}
\sum_{r=0}^{N-d}\frac{\binom{N-d}{r}}{\binom Nr}
&=(N+1)\int_0^1(1-t)^d
\sum_{r=0}^{N-d}\binom{N-d}{r}t^r(1-t)^{N-d-r}\,dt\\
&=(N+1)\int_0^1(1-t)^d\,dt
=\frac{N+1}{d+1}.
\end{align*}
Substituting $N=n-1$ proves the claim.
\end{proof}

\begin{theorem}[Tree-distance identity]\label{thm:distance-identity}
\begin{equation}
G_n=\frac{n^2}{4}\E\left[\frac{1}{D+1}\right].
\label{eq:distance-identity}
\end{equation}
\end{theorem}

\begin{proof}
From the conditional identity preceding Lemma~\ref{lem:connectivity-sum},
$\E[Q_k]=n^2\Pr(U\leftrightarrow V\text{ at level }k)$. Hence
\begin{align*}
G_n
&=\frac{1}{4n}\sum_{k=1}^n\E[Q_k]\\
&=\frac n4\E\!\left[
  \sum_{k=1}^n
  \Pr(U\leftrightarrow V\text{ at level }k\mid D)
  \right]\\
&=\frac n4\E\!\left[\frac n{D+1}\right],
\end{align*}
where the last equality is Lemma~\ref{lem:connectivity-sum}. This is
\eqref{eq:distance-identity}.
\end{proof}

The distance distribution below is classical in random-tree enumeration
\citep{meir1970distance}; the short count is included to fix normalization and
cover the full range of $d$.

\begin{lemma}[Distance distribution in a Cayley tree]\label{lem:distance-law}
For two fixed distinct vertices $u,v$ of a uniform Cayley tree and
$1\leq d\leq n-1$,
\begin{equation}
\Pr(\dist_T(u,v)=d)
=\frac{(d+1)(n-2)!}{(n-d-1)!\,n^d}.
\label{eq:distance-law}
\end{equation}
\end{lemma}

\begin{proof}
The $u$--$v$ path in a tree is unique. Choose and order its $d-1$ internal
vertices in $(n-2)!/(n-d-1)!$ ways. Once one such path is fixed, regard it as
one forest component of size $d+1$ and every remaining vertex as a singleton.
Lemma~\ref{lem:forest-extension} shows that the number of Cayley trees
containing this path is $(d+1)n^{n-d-2}$. Every such tree has the selected path
as its unique $u$--$v$ path, so the counts for different ordered internal
vertex lists are disjoint. This reasoning also covers $d=n-1$: the fixed path
then spans all vertices and its unique completion count is
$(d+1)n^{-1}=1$. Dividing the product of the two counts by Cayley's total
$n^{n-2}$ gives \eqref{eq:distance-law}.
\end{proof}

\begin{theorem}[Exact finite-size formula]\label{thm:exact-formula}
For every integer $n\geq3$,
\begin{align}
G_n
&=\frac n4\left[
1+\sum_{d=1}^{n-1}
\frac{(n-1)!}{(n-d-1)!\,n^d}
\right]
\label{eq:exact-factorial}\\
&=\frac n4\sum_{d=0}^{n-1}
\prod_{j=1}^{d}\left(1-\frac jn\right),
\label{eq:exact-product}
\end{align}
where the empty product equals one.
\end{theorem}

\begin{proof}
The event $U=V$ has probability $1/n$ and gives $D=0$. Conditional on
$U\neq V$, Lemma~\ref{lem:distance-law} applies. Hence
\[
\E\left[\frac1{D+1}\right]
=\frac1n+
\frac{n-1}{n}\sum_{d=1}^{n-1}
\frac{(n-2)!}{(n-d-1)!\,n^d}.
\]
Substitute this expression into \eqref{eq:distance-identity} to obtain
\eqref{eq:exact-factorial}. The identity
\[
\frac{(n-1)!}{(n-d-1)!\,n^d}
=\prod_{j=1}^d\left(1-\frac jn\right)
\]
gives \eqref{eq:exact-product}.
\end{proof}

\section{Asymptotic Growth}\label{sec:asymptotic}

We now extract the exact leading order from
\eqref{eq:exact-product}. Define
\[
a_{n,d}=\prod_{j=1}^{d}\left(1-\frac jn\right),
\qquad a_{n,0}=1.
\]

\begin{theorem}[Asymptotic greedy cost]\label{thm:asymptotic}
\begin{equation}
G_n\sim\frac14\sqrt{\frac{\pi}{2}}\,n^{3/2}.
\label{eq:asymptotic}
\end{equation}
\end{theorem}

\begin{proof}
Fix $M>0$. Uniformly for $0\leq d\leq M\sqrt n$, the Taylor estimate
$\log(1-u)=-u+O(u^2)$ gives
\begin{align}
\log a_{n,d}
&=-\frac1n\sum_{j=1}^d j
+O\left(\frac1{n^2}\sum_{j=1}^d j^2\right)\notag\\
&=-\frac{d(d+1)}{2n}+O_M(n^{-1/2}).
\label{eq:local-log}
\end{align}
For $d=\lfloor x\sqrt n\rfloor$, the first term on the right converges to
$-x^2/2$ uniformly for $x\in[0,M]$. Exponentiation on the resulting bounded
interval therefore gives the explicit compact convergence
\[
\sup_{0\leq x\leq M}
\left|a_{n,\lfloor x\sqrt n\rfloor}-e^{-x^2/2}\right|
\longrightarrow0.
\]
The Riemann-sum consequence is
\begin{equation}
\frac1{\sqrt n}\sum_{0\leq d\leq M\sqrt n}a_{n,d}
\longrightarrow \int_0^M e^{-x^2/2}\,dx.
\label{eq:compact-riemann}
\end{equation}

For every $d<n$, $\log(1-u)\leq-u$ gives
\begin{equation}
0\leq a_{n,d}
\leq \exp\left(-\frac{d(d+1)}{2n}\right)
\leq e^{-d^2/(2n)}.
\label{eq:gaussian-domination}
\end{equation}
Writing $m_n=\lfloor M\sqrt n\rfloor+1$ and using that the Gaussian density is
decreasing on $[0,\infty)$, \eqref{eq:gaussian-domination} yields
\begin{align*}
0&\leq\frac1{\sqrt n}\sum_{d=m_n}^{n-1}a_{n,d}
\leq\frac1{\sqrt n}\sum_{d=m_n}^{\infty}e^{-d^2/(2n)}\\
&\leq\frac1{\sqrt n}+
\int_{m_n/\sqrt n}^{\infty}e^{-x^2/2}\,dx.
\end{align*}
Thus the limsup of the scaled tail is at most
$\int_M^\infty e^{-x^2/2}\,dx$. Combining this inequality with
\eqref{eq:compact-riemann} gives
\begin{align*}
\limsup_{n\to\infty}
\left|\frac1{\sqrt n}\sum_{d=0}^{n-1}a_{n,d}
-\int_0^\infty e^{-x^2/2}\,dx\right|
\leq 2\int_M^\infty e^{-x^2/2}\,dx.
\end{align*}
The right-hand side tends to zero as $M\to\infty$, so
\[
\frac1{\sqrt n}\sum_{d=0}^{n-1}a_{n,d}
\longrightarrow
\int_0^\infty e^{-x^2/2}\,dx
=\sqrt{\frac{\pi}{2}}.
\]
Multiplication by $n/4$ proves \eqref{eq:asymptotic}.
\end{proof}

\begin{corollary}\label{cor:per-arrival}
\[
G_n=\Theta(n^{3/2}),
\qquad
\frac{G_n}{n}\sim
\frac14\sqrt{\frac{\pi}{2}}\,\sqrt n.
\]
\end{corollary}

The leading constant is
\[
\frac14\sqrt{\frac{\pi}{2}}\approx0.3133285343.
\]

\section{Stochastic Performance Benchmarks}\label{sec:benchmarks}

Let $\OPT_n^{\mathrm{off}}$ be the minimum total matching cost after all
$n$ arrivals are revealed, and let
\[
O_n=\E[\OPT_n^{\mathrm{off}}].
\]
The following result compares expected costs under the fixed stochastic model.
It is not a worst-case competitive-ratio statement.

\begin{theorem}[Offline lower bound]\label{thm:offline-lower}
\begin{equation}
\liminf_{n\to\infty}\frac{O_n}{n^{3/2}}
\geq\frac1{\sqrt{24\pi}}.
\label{eq:offline-lower}
\end{equation}
\end{theorem}

\begin{proof}
Fix vertex $0$ and set $f(x)=d(x,0)$. This function is $1$-Lipschitz on the
cycle. For any bijection $\sigma$ matching arrivals $X_i$ to the facility
vertices,
\begin{align}
\sum_{i=1}^n d(X_i,\sigma(i))
&\geq\sum_{i=1}^n|f(X_i)-f(\sigma(i))|\notag\\
&\geq\left|
\sum_{i=1}^n f(X_i)-\sum_{v\in\V}f(v)
\right|.
\label{eq:lipschitz-certificate}
\end{align}

For a uniform point $X$ on the length-$n$ cycle,
$f(X)$ is uniform on $[0,n/2]$. Thus
\[
\E f(X)=\frac n4,
\qquad
\operatorname{Var}(f(X))=\frac{n^2}{48}.
\]
The deterministic facility sum equals $n^2/4$ for even $n$ and
$(n^2-1)/4$ for odd $n$. Its difference from $n\E f(X)=n^2/4$ is therefore
bounded by $1/4$.

Let
\[
Z_n=\frac{\sum_{i=1}^n(f(X_i)-n/4)}
{n^{3/2}/\sqrt{48}}.
\]
After multiplication by $\sqrt{48}/n$, each centered summand has a
distribution independent of $n$, with mean zero and variance one. The
classical central limit theorem therefore gives $Z_n\Rightarrow Z$, where
$Z$ is standard normal. Moreover, $\E[Z_n^2]=1$, so
$\{|Z_n|\}$ is uniformly integrable. The bounded centering discrepancy above
is $o(n^{3/2})$, and hence
\[
\E\left|
\sum_{i=1}^n f(X_i)-\sum_{v\in\V}f(v)
\right|
=\frac{n^{3/2}}{\sqrt{48}}
\left(\E|Z|+o(1)\right).
\]
Since $\E|Z|=\sqrt{2/\pi}$, taking the minimum over $\sigma$ in
\eqref{eq:lipschitz-certificate} and then expectation proves
\eqref{eq:offline-lower}.
\end{proof}

\begin{corollary}[Stochastic ratio of expectations]\label{cor:stochastic-ratio}
\begin{equation}
\limsup_{n\to\infty}\frac{G_n}{O_n}
\leq\frac{\pi\sqrt3}{2}.
\label{eq:stochastic-ratio}
\end{equation}
If $V_n^\star$ is the expected cost of the optimal online policy under the
same arrival law, then
\begin{equation}
\limsup_{n\to\infty}\frac{G_n}{V_n^\star}
\leq\frac{\pi\sqrt3}{2}.
\label{eq:online-ratio}
\end{equation}
\end{corollary}

\begin{proof}
Theorem~\ref{thm:asymptotic} and the positive liminf bound in
Theorem~\ref{thm:offline-lower} imply
\[
\limsup_{n\to\infty}\frac{G_n}{O_n}
\leq
\frac{\frac14\sqrt{\pi/2}}{1/\sqrt{24\pi}}.
\]
The constants simplify as
\[
\frac{\frac14\sqrt{\pi/2}}{1/\sqrt{24\pi}}
=\frac{\pi\sqrt3}{2}.
\]
Finally, $O_n\leq V_n^\star\leq G_n$, so replacing $O_n$ by the larger
denominator $V_n^\star$ preserves the upper bound.
\end{proof}

\section{Optimal Stochastic Online Assignment}\label{sec:optimal-online}

Greedy minimizes only the current distance. To quantify this restriction, let
$V^\star(S)$ be the minimum expected remaining cost over all online policies
from occupied set $S$.

\begin{proposition}[Optimal-online Bellman equation]\label{prop:bellman}
\begin{equation}
V^\star(S)=\frac1n\sum_{e=0}^{n-1}\int_0^1
\min_{v\notin S}
\left\{d(e,t;v)+V^\star(S\cup\{v\})\right\}\,dt,
\label{eq:bellman}
\end{equation}
with $V^\star(\V)=0$.
\end{proposition}

\begin{proof}
The future arrivals are independent of the entire observed history and have
the same uniform law. Capacities are unit and costs are additive, so the
occupied set $S$ is a sufficient state: conditional on $S$, no other feature
of the history changes the feasible actions or the future distribution. After
observing the current point $(e,t)$, the finite action set is
$\V\setminus S$, and choosing $v$ incurs immediate cost $d(e,t;v)$ followed by
optimal continuation value $V^\star(S\cup\{v\})$. Pointwise minimization over
that finite set and then averaging over the uniform arrival gives
\eqref{eq:bellman}. Since every action increases $|S|$ by one, backward
induction from $V^\star(\V)=0$ proves existence and optimality of the recursion.
\end{proof}

The integral in \eqref{eq:bellman} has a particularly simple exact form.
Define
\begin{align}
A_e(S)&=\min_{v\notin S}
\{\delta(e,v)+V^\star(S\cup\{v\})\},\notag\\
B_e(S)&=\min_{v\notin S}
\{\delta(e+1,v)+V^\star(S\cup\{v\})\}.
\label{eq:AB-values}
\end{align}
For $c_v=V^\star(S\cup\{v\})$, \eqref{eq:point-distance} and the elementary
finite-set identity
\[
\min_v\min\{a_v+c_v,b_v+c_v\}
=\min\{\min_v(a_v+c_v),\min_v(b_v+c_v)\}
\]
give
\[
\min_{v\notin S}\{d(e,t;v)+V^\star(S\cup\{v\})\}
=\min\{A_e(S)+t,\,B_e(S)+1-t\}.
\]
For constants $A,B$, put
\[
\tau=\frac{1+B-A}{2}
\]
and define
\begin{equation}
h(A,B)=\int_0^1\min\{A+t,B+1-t\}\,dt.
\label{eq:h-definition}
\end{equation}
Direct integration gives
\begin{equation}
h(A,B)=
\begin{cases}
B+\frac12, & \tau\leq0,\\[2mm]
A\tau+\tau^2+(B+1)(1-\tau)-\frac12,
&0<\tau<1,\\[2mm]
A+\frac12, & \tau\geq1.
\end{cases}
\label{eq:h-closed}
\end{equation}
Consequently,
\begin{equation}
V^\star(S)=\frac1n\sum_{e=0}^{n-1}
h(A_e(S),B_e(S)).
\label{eq:rational-bellman}
\end{equation}
Starting from $V^\star(\V)=0$, every quantity in
\eqref{eq:rational-bellman} is rational. Thus the optimal values for small
$n$ can be evaluated without numerical quadrature.

\begin{proposition}[The triangle]\label{prop:n3-optimal}
For $n=3$, nearest-available greedy is optimal from the initial state, and
\[
V_3^\star=G_3=\frac{17}{12}.
\]
\end{proposition}

\begin{proof}
By rotational symmetry, the Bellman values depend only on the number of
occupied vertices. Starting at the terminal value zero, the three orbit
values with two, one, and zero occupied vertices are respectively
\[
\frac34,\qquad \frac76,\qquad \frac{17}{12}.
\]
For example, at $S=\{0\}$ the three edge integrals from
\eqref{eq:rational-bellman} are $5/4,1,5/4$, whose average is $7/6$; at the
empty state every edge integral is $17/12$. These substitutions constitute an
exact Bellman certificate. Theorem~\ref{thm:exact-formula} independently gives
$G_3=17/12$. Since greedy attains the Bellman lower value, it is optimal from
the initial state.
\end{proof}

\begin{proposition}[Greedy is not optimal]\label{prop:not-optimal}
For $n=4$,
\[
V_4^\star=\frac{1135}{512},
\qquad
G_4=\frac{71}{32}=\frac{1136}{512}.
\]
Hence greedy exceeds the optimal online expected cost by $1/512$.
\end{proposition}

\begin{proof}
Equation~\eqref{eq:rational-bellman}, evaluated backward on the dihedral
state orbits of the square, gives the following orbit values.
\begin{center}
\begin{tabular}{c c c}
\toprule
Occupied facilities & Representative & $V^\star(S)$\\
\midrule
4 & $\{0,1,2,3\}$ & $0$\\
3 & $\{0,1,2\}$ & $1$\\
2, opposite & $\{0,2\}$ & $3/2$\\
2, adjacent & $\{0,1\}$ & $13/8$\\
1 & $\{0\}$ & $1007/512$\\
0 & $\emptyset$ & $1135/512$\\
\bottomrule
\end{tabular}
\end{center}
\ref{app:n4-certificate} gives every
$A_e(S),B_e(S),\tau$, and $h(A_e(S),B_e(S))$ substitution for these orbit
representatives, so the displayed fractions can be checked without trusting
software.

The source of strict improvement is also explicit. At state $S=\{0\}$, on
edge $e=1$ the two relevant cost-plus-continuation lines cross at
$\tau=7/16$, while nearest-available greedy switches facilities at $t=1/2$.
On edge $e=2$ the corresponding thresholds are $9/16$ and $1/2$. On each of
the two discrepant intervals the integrated improvement is the area of a
triangle of base $1/16$ and height $1/8$, namely $1/256$. Averaging the two
improvements over the four equally likely edges gives $1/512$. At the empty
state all continuation values after the first assignment are equal by
rotation, so the Bellman action coincides almost everywhere with greedy and
the same $1/512$ improvement propagates to the initial value.
Theorem~\ref{thm:exact-formula} gives $G_4=71/32=1136/512$, while the exact
certificate gives $V_4^\star=1135/512$.
\end{proof}

Further exact values for $3\leq n\leq7$ appear in
\ref{app:validation}. They do not establish asymptotic optimality.

\section{Computation and reproducibility}\label{sec:computation}

The exact mean can be evaluated by
\[
a_{n,0}=1,\qquad a_{n,d}=a_{n,d-1}(1-d/n),\qquad
G_n=\frac n4\sum_{d=0}^{n-1}a_{n,d}.
\]
This needs a linear number of scalar steps. Exact rational evaluation has
additional bit costs. The gap and occupied-state recurrences in
\ref{app:recurrences} provide independent computations for small
instances. The supplementary verifier checks their agreement with the
product formula for $3\leq n\leq8$, the rational Bellman calculations, the
Cayley distance law through $n=6$, and the integrated connectivity identity
through $n=5$ by exhaustive enumeration.

Supplementary Material~1, \texttt{TCS\_supplement.zip}, contains the
standard-library Python verifier \texttt{code/verify\_ofa.py}, generated CSV
tables, LaTeX table fragments, and reproduction instructions. The deterministic
checks run with \texttt{python code/verify\_ofa.py --test}. Numerical tables
and simulation checks are in \ref{app:validation}, including the
correction to the earlier large-instance estimates. The exact formula
and asymptotic theorem do not depend on simulation.

\section{Discussion}\label{sec:discussion}

The gap representation explains the growth mechanism directly. The immediate
cost is proportional to the sum of squared gaps. Under the additive
coalescent, a few large clusters emerge, causing this second moment to be much
larger than its balanced lower bound $n^2/k$. The Cayley-tree representation
integrates that susceptibility over the complete process and reveals the
$\sqrt n$ scale of a typical path distance that determines the final
$n^{3/2}$ cost.

The stochastic comparison must be interpreted carefully. Corollary
\ref{cor:stochastic-ratio} is an asymptotic ratio-of-expectations guarantee
under the specified model of independent uniform arrivals. It is not a
sequence-by-sequence competitive ratio and says nothing about adversarial
inputs. Similarly, the small values in
\ref{app:validation} show that greedy is close to optimal for those
instances but do not establish $G_n/V_n^\star\to1$.

The analysis relies on four defining structural assumptions of the model:
equal facility spacing, unit capacities, uniform independent arrivals, and
cycle distance. With nonuniform arrival density, the one-step gap cost and transition weights
change. With heterogeneous capacities, removing a facility is no longer a
single coalescence event. These variants require separate analysis and are not
covered by the present theorem.

\subsection{Open Problems}

The present analysis leaves several natural questions concerning stochastic
online assignment on cycle metrics.

\begin{enumerate}
\item \emph{Asymptotics of the optimal online policy.}
Determine the first-order asymptotic behavior of $V_n^\star$, the expected cost
of the optimal stochastic online policy. In particular, does
\[
\lim_{n\to\infty}\frac{G_n}{V_n^\star}
\]
exist, and, if so, what is its value? The exact computations for small $n$
reported here do not resolve this question. A related algorithmic question is whether the cycle structure permits a polynomial-time exact computation of $V_n^\star$, a compact structural
recurrence, or an approximation scheme with a rigorous error guarantee.

\item \emph{The sharp offline comparison.}
Let $O_n=\E[\OPT_n^{\mathrm{off}}]$. Theorem~\ref{thm:offline-lower}, together
with $O_n\leq G_n$, determines the order $n^{3/2}$ but not the leading constant
for $O_n$. It remains to determine whether $O_n/n^{3/2}$ converges and, if it
does, to identify its limit. This would also determine the sharp asymptotic
ratio between greedy and the expected offline optimum.

\item \emph{Beyond the uniform, equally spaced unit-capacity model.}
For nonuniform arrival densities, the one-step gap cost and merger
probabilities are no longer functions only of integer gap sizes. Unequal
facility spacing introduces additional geometric information, while
heterogeneous capacities prevent every assignment from producing a single
merger. It would be useful to identify broader model classes admitting a
finite-dimensional Markov or coalescent representation and to determine their
corresponding greedy-cost asymptotics.

\item \emph{Other metric spaces.}
A further question is whether related coalescent or random-tree
representations can be obtained for paths, trees, or other highly symmetric
metrics. Boundary effects on a path and branching on a tree destroy the cyclic
exchangeability used here, so new state descriptions would be required.

\end{enumerate}

\section{Conclusion}\label{sec:conclusion}

We obtained an exact and asymptotic average-case analysis of nearest-available
greedy OFA on a cycle. The free-facility gaps form an adjacent size-biased
merging process; after an exchangeability argument, their unordered masses are
the discrete additive coalescent. A uniform-Cayley-tree representation then
reduces the total expected assignment cost to a two-vertex distance functional
and yields the exact formula
\[
G_n=\frac n4\sum_{d=0}^{n-1}
\prod_{j=1}^{d}\left(1-\frac jn\right).
\]
Its rigorous asymptotic evaluation gives
$G_n\sim\frac14\sqrt{\pi/2}\,n^{3/2}$.

Under the specified stochastic model, the offline comparison is precisely the
asymptotic ratio-of-expectations guarantee
\[
\limsup_{n\to\infty}
\frac{G_n}{\E[\OPT_n^{\mathrm{off}}]}
\leq\frac{\pi\sqrt3}{2}.
\]
We also computed the optimal stochastic online policy for small cycles. The
square already separates the two policies: greedy pays exactly $1/512$ more in
expectation.

Open problems include sharper stochastic comparison constants, asymptotics for
the optimal online policy, nonuniform arrival distributions, and extensions
with heterogeneous capacities. 

\appendix
\numberwithin{equation}{section}
\numberwithin{table}{section}
\numberwithin{figure}{section}
\renewcommand{\theequation}{\Alph{section}.\arabic{equation}}
\renewcommand{\thetable}{\Alph{section}.\arabic{table}}
\renewcommand{\thefigure}{\Alph{section}.\arabic{figure}}
\section{Exact recurrences for independent validation}\label{app:recurrences}
\subsection{Occupied-state recurrence}

For completeness, we record the occupied-state recurrence used in the
earlier study \citep{riad2026isco}. Let
\[
N_S(e,t)=\arg\min_{v\in F(S)}d(e,t;v)
\]
and let $J(S)$ be the expected remaining greedy cost.

\begin{proposition}[Occupied-state recurrence]\label{prop:occupied-recurrence}
For $S\neq\V$,
\begin{equation}
J(S)=\frac1n\sum_{e=0}^{n-1}\int_0^1
\frac1{|N_S(e,t)|}\sum_{v\in N_S(e,t)}
\bigl(d(e,t;v)+J(S\cup\{v\})\bigr)\,dt,
\label{eq:occupied-recurrence}
\end{equation}
with $J(\V)=0$. In particular, $G_n=J(\emptyset)$.
\end{proposition}

\begin{proof}
Condition on the next arrival point and its tie-breaking outcome. The immediate
cost is $d(e,t;v)$ and the successor state is $S\cup\{v\}$. Averaging gives
\eqref{eq:occupied-recurrence}. Each transition increases $|S|$, so backward
induction gives existence and uniqueness.
\end{proof}

Rotations and reflections of $C_n$ preserve the distance, arrival law, and
greedy rule. Hence $J(S)=J(\gamma S)$ for every dihedral symmetry $\gamma$.
This permits exact dynamic programming over one representative of each
rotation-reflection orbit. It is an exact computational acceleration, but the
gap process in Section~\ref{sec:gaps} is the main analytical representation.

\subsection{Gap-state recurrence}

Let $M_i\mathbf g$ denote the cyclic vector after this merger and let
$H(\mathbf g)$ be the expected remaining greedy cost from gap state
$\mathbf g$.

\begin{corollary}[Gap recurrence]\label{cor:gap-recurrence}
For $k\geq2$,
\begin{equation}
H(\mathbf g)=\frac{1}{4n}\sum_{j=1}^k g_j^2
+\sum_{i=1}^k\frac{g_{i-1}+g_i}{2n}H(M_i\mathbf g),
\label{eq:gap-recurrence}
\end{equation}
and $H((n))=n/4$. The desired cost is
$G_n=H((1,\ldots,1))$.
\end{corollary}

\begin{proof}
Combine Theorems~\ref{thm:gap-cost} and~\ref{thm:gap-transition}. With one
free facility, the expected distance of a uniform point on a length-$n$ circle
to that facility is $n/4$, which also agrees with
\eqref{eq:gap-cost}.
\end{proof}

The gap formula already gives a general lower bound. When $k$ gaps remain,
Cauchy--Schwarz and \eqref{eq:gap-composition} give
$\sum_i g_i^2\geq n^2/k$. Summing the resulting one-step bounds yields
\begin{equation}
G_n\geq\frac n4\sum_{k=1}^n\frac1k
=\frac n4H_n.
\label{eq:harmonic-lower}
\end{equation}
This lower bound is not tight; the exact order is $n^{3/2}$.

\section{Exact computations and validation}\label{app:validation}
\suppressfloats[t]

The occupied-state recurrence has one value for each subset of facilities,
hence $2^n$ states before symmetry reduction. By contrast, the product formula
\eqref{eq:exact-product} can be evaluated using $n-1$ recurrence updates and
one length-$n$ sum:
\[
a_{n,0}=1,\qquad
a_{n,d}=a_{n,d-1}\left(1-\frac dn\right),
\qquad
G_n=\frac n4\sum_{d=0}^{n-1}a_{n,d}.
\]
Thus the number of scalar recurrence steps is linear in $n$. For exact
rational arithmetic, we do not claim linear bit complexity.

For independent checks, we evaluated both the gap recurrence
\eqref{eq:gap-recurrence} and the occupied-set recurrence
\eqref{eq:occupied-recurrence} using exact rational arithmetic. Gap states
were canonicalized under rotation and reflection, while occupied states were
integrated directly from their piecewise-affine distance functions.
Table~\ref{tab:exact-small} shows agreement of all three routes with
Theorem~\ref{thm:exact-formula}.

\begin{table}[htbp]
\centering
\caption{Exact greedy costs. The product formula, cyclic-gap recurrence, and
occupied-set integration agree on every displayed instance.}
\label{tab:exact-small}
\begin{tabular}{c c c}
\toprule
$n$ & Exact $G_n$ & Decimal\\
\midrule
3 & $17/12$ & 1.416667\\
4 & $71/32$ & 2.218750\\
5 & $1569/500$ & 3.138000\\
6 & $899/216$ & 4.162037\\
7 & $355081/67228$ & 5.281743\\
8 & $425331/65536$ & 6.490036\\
9 & $16541017/2125764$ & 7.781210\\
\bottomrule
\end{tabular}
\end{table}

Table~\ref{tab:large-exact} gives the corrected large-instance values. The final column approaches the constant in
Theorem~\ref{thm:asymptotic}.

\begin{table}[htbp]
\centering
\caption{Large-instance values evaluated from the exact product recurrence;
all displayed entries are rounded only after evaluation.}
\label{tab:large-exact}
\begin{tabular}{r r r r}
\toprule
$n$ & $G_n$ & $G_n/n$ & $G_n/n^{3/2}$\\
\midrule
20   & 26.467923   & 1.323396  & 0.295920\\
50   & 106.789088  & 2.135782  & 0.302045\\
100  & 305.249016  & 3.052490  & 0.305249\\
500  & 3462.029339 & 6.924059  & 0.309653\\
1000 & 9825.803232 & 9.825803  & 0.310719\\
\bottomrule
\end{tabular}
\end{table}

\subsection{Correction of the earlier numerical values}\label{app:numerical-correction}
The earlier study \citep{riad2026isco} reported large-instance values
obtained by simulation. After deriving Theorem~\ref{thm:exact-formula}, we found that those
values were inconsistent both with the exact analysis and with an independently
written direct simulation of the stochastic process. We have not identified
the cause in the original implementation and therefore do not attribute the
discrepancy to a particular bug. As a reproducible check, 50,000 independent
runs for each size, using base pseudorandom seed $20260814$ and seed
$20260814+n$ for size $n$, gave sample means and normal-approximation 95\%
intervals $26.502388\ [26.438695,26.566082]$ for $n=20$,
$106.702509\ [106.457516,106.947502]$ for $n=50$, and
$305.283159\ [304.598637,305.967680]$ for $n=100$. Each interval contains the
corresponding product-formula value, rounded in
Table~\ref{tab:large-exact}. All large-instance
values reported in the present paper are rounded evaluations of the exact
product formula.

\begin{table}[htbp]
\centering
\caption{Independent direct-simulation check. Intervals are
normal-approximation 95\% confidence intervals; exact means come from the
product formula.}
\label{tab:simulation-check}
\setlength{\tabcolsep}{4pt}
\begin{tabular}{r r r r c r}
\toprule
$n$ & runs & seed & mean & 95\% CI & exact mean\\
\midrule
20 & 50,000 & 20260834 & 26.502388 & $[26.438695,26.566082]$ & 26.467923\\
50 & 50,000 & 20260864 & 106.702509 & $[106.457516,106.947502]$ & 106.789088\\
100 & 50,000 & 20260914 & 305.283159 & $[304.598637,305.967680]$ & 305.249016\\
\bottomrule
\end{tabular}
\end{table}

The occupied-set recurrence, the gap-state recurrence, and the product formula
are three mathematically distinct computational routes. Agreement among them
on the finite instances above checks the indexing, distance convention, and
terminal step. A source-distribution verification script evaluates the exact
formula and the optimal Bellman recurrence using rational arithmetic and also
provides the independent direct simulator used above.

\subsection{Reproducibility and independent checks}\label{sec:reproducibility}
Supplementary Material~1 contains the standard-library Python program
\texttt{code/verify\_ofa.py}, the generated tables in machine-readable
CSV form, and an expected-results certificate. Exact computations use
\texttt{fractions.Fraction}; no floating-point quadrature enters any claimed
exact value. Running \texttt{--test} performs the following checks:
\begin{enumerate}
\item equality of the product, gap-state, and independently integrated
occupied-state routes for $3\leq n\leq8$;
\item the exact Bellman values and complete square certificate;
\item the closed form \eqref{eq:h-closed} against independent piecewise-affine
integration;
\item the Cayley-tree distance law by enumeration of every Pr\"ufer sequence
through $n=6$; and
\item the susceptibility/connectivity identity by enumerating every labeled
tree and every edge order for $3\leq n\leq5$.
\end{enumerate}
The reported simulation was reproduced with CPython~3.12.14,
\texttt{random.Random(20260814+n)}, 50,000 independent runs for each $n$, the
sample standard deviation with denominator $R-1$, and multiplier $1.96$ for
the displayed intervals. The exact claims do not depend on the pseudorandom
generator or on those intervals.

\subsection{Additional optimal-online values}

Table~\ref{tab:optimal-small} gives further exact-rational evaluations,
rounded only for display. These data do not imply that greedy is
asymptotically optimal.

\begin{table}[htbp]
\centering
\caption{Greedy versus the optimal stochastic online policy. Both columns are
obtained by exact rational recurrences; displayed decimals are rounded. The
last column is $100(G_n/V_n^\star-1)$.}
\label{tab:optimal-small}
\begin{tabular}{c c c c}
\toprule
$n$ & $G_n$ & $V_n^\star$ & Greedy excess (\%)\\
\midrule
3 & 1.416667 & 1.416667 & 0.000\\
4 & 2.218750 & 2.216797 & 0.088\\
5 & 3.138000 & 3.130760 & 0.231\\
6 & 4.162037 & 4.145260 & 0.405\\
7 & 5.281743 & 5.250560 & 0.594\\
\bottomrule
\end{tabular}
\end{table}

\section{Exact Bellman certificate for the square}\label{app:n4-certificate}

Table~\ref{tab:n4-certificate} lists every substitution needed to verify
Proposition~\ref{prop:not-optimal}. The sets shown are representatives of all
dihedral occupied-set orbits. Values for other sets follow by rotating or
reflecting the square. Within each row block, averaging the last column over
the four edges gives the Bellman value for that representative; those values
are then used in the next block below it.

\small
\begin{longtable}{c r r r r r}
\caption{Complete edge-integral certificate for the square. For each occupied
set $S$ and edge $e$, the table gives $A_e(S)$, $B_e(S)$, the crossing
$\tau=(1+B_e-A_e)/2$, and $h(A_e,B_e)$.}
\label{tab:n4-certificate}\\
\toprule
$S$ & $e$ & $A_e$ & $B_e$ & $\tau$ & $h(A_e,B_e)$\\
\midrule
\endfirsthead
\multicolumn{6}{c}{\tablename\ \thetable\ (continued)}\\
\toprule
$S$ & $e$ & $A_e$ & $B_e$ & $\tau$ & $h(A_e,B_e)$\\
\midrule
\endhead
\midrule
\multicolumn{6}{r}{Continued on next page}\\
\endfoot
\bottomrule
\endlastfoot
$\{0,1,2,3\}$ & -- & -- & -- & -- & $0$\\
$\{0,1,2\}$ & 0 & $1$ & $2$ & $1$ & $\frac{3}{2}$\\
$\{0,1,2\}$ & 1 & $2$ & $1$ & $0$ & $\frac{3}{2}$\\
$\{0,1,2\}$ & 2 & $1$ & $0$ & $0$ & $\frac{1}{2}$\\
$\{0,1,2\}$ & 3 & $0$ & $1$ & $1$ & $\frac{1}{2}$\\
\addlinespace
$\{0,1\}$ & 0 & $2$ & $2$ & $\frac{1}{2}$ & $\frac{9}{4}$\\
$\{0,1\}$ & 1 & $2$ & $1$ & $0$ & $\frac{3}{2}$\\
$\{0,1\}$ & 2 & $1$ & $1$ & $\frac{1}{2}$ & $\frac{5}{4}$\\
$\{0,1\}$ & 3 & $1$ & $2$ & $1$ & $\frac{3}{2}$\\
\addlinespace
$\{0,2\}$ & 0 & $2$ & $1$ & $0$ & $\frac{3}{2}$\\
$\{0,2\}$ & 1 & $1$ & $2$ & $1$ & $\frac{3}{2}$\\
$\{0,2\}$ & 2 & $2$ & $1$ & $0$ & $\frac{3}{2}$\\
$\{0,2\}$ & 3 & $1$ & $2$ & $1$ & $\frac{3}{2}$\\
\addlinespace
$\{0\}$ & 0 & $\frac{21}{8}$ & $\frac{13}{8}$ & $0$ & $\frac{17}{8}$\\
$\{0\}$ & 1 & $\frac{13}{8}$ & $\frac{3}{2}$ & $\frac{7}{16}$ & $\frac{463}{256}$\\
$\{0\}$ & 2 & $\frac{3}{2}$ & $\frac{13}{8}$ & $\frac{9}{16}$ & $\frac{463}{256}$\\
$\{0\}$ & 3 & $\frac{13}{8}$ & $\frac{21}{8}$ & $1$ & $\frac{17}{8}$\\
\addlinespace
$\varnothing$ & 0 & $\frac{1007}{512}$ & $\frac{1007}{512}$ & $\frac{1}{2}$ & $\frac{1135}{512}$\\
$\varnothing$ & 1 & $\frac{1007}{512}$ & $\frac{1007}{512}$ & $\frac{1}{2}$ & $\frac{1135}{512}$\\
$\varnothing$ & 2 & $\frac{1007}{512}$ & $\frac{1007}{512}$ & $\frac{1}{2}$ & $\frac{1135}{512}$\\
$\varnothing$ & 3 & $\frac{1007}{512}$ & $\frac{1007}{512}$ & $\frac{1}{2}$ & $\frac{1135}{512}$\\
\end{longtable}
\normalsize

\section*{Funding}
This research did not receive any specific grant from funding agencies in
the public, commercial, or not-for-profit sectors.

\section*{Declaration of competing interests}
The authors declare that they have no known competing financial interests
or personal relationships that could have appeared to influence the work
reported in this paper.

\section*{CRediT authorship contribution statement}
\textbf{Md. Tanzeem Rahat:} Conceptualization, Methodology, Formal analysis,
Writing -- original draft, Writing -- review and editing.
\textbf{Md. Rawha Siddiqi Riad:} Conceptualization, Methodology, Software,
Formal analysis, Investigation, Writing -- original draft.
\textbf{Md. Manzurul Hasan:} Conceptualization, Supervision, Methodology,
Formal analysis.

\section*{Data and code availability}
No external datasets were used. The code and generated numerical data are
included in Supplementary Material~1. Repository record:
\texttt{[REPOSITORY URL OR DOI TO BE ADDED BEFORE SUBMISSION]}.

\section*{Declaration of generative AI and AI-assisted technologies in the manuscript preparation process}
During preparation and revision of this manuscript, the authors used OpenAI
Codex to assist with language editing, manuscript organization, comparison
with related literature, LaTeX preparation, and reproducibility checks.
The authors are responsible for reviewing and verifying the resulting text,
references, mathematical arguments, and code, and take full responsibility
for the content of the article.

\section*{Acknowledgements}
The authors acknowledge the American International University-Bangladesh
(AIUB) for institutional support. The authors also thank Reyan Ahmed,
Department of Computer Science, University of Arizona, for valuable initial
discussions and for suggesting research directions that helped shape the
formulation of this work.

\bibliographystyle{elsarticle-num-names}
\bibliography{references}
\end{document}